\documentclass[manuscript,screen,nonacm]{acmart}

\usepackage{amsmath, amsfonts}
\usepackage{dsfont}
\usepackage{subcaption}
\usepackage{comment}

\newcommand{\p}{\mathbb{P}}
\newcommand{\majority}{\text{majority}}
\newcommand{\minority}{\text{minority}}
\newcommand{\doop}[1]{\text{do}(#1)}
\newtheorem{hypothesis}{Hypothesis}
\newtheorem{test}{Empirical Test}
\newcommand{\FIGURE}[3]{\centering #1\caption{#2}#3}
\newcommand{\TABLE}[3]{\centering\caption{#1}#2#3}

\begin{document}

\title{Fairness Is More Than Algorithms: Racial Disparities in Time-to-Recidivism}

\author{Jessy Xinyi Han}
\affiliation{%
  \institution{Massachusetts Institute of Technology}
  \city{Cambridge, MA}
  \country{USA}}
\email{xyhan@mit.edu}

\author{Kristjan Greenewald}
\affiliation{%
  \institution{MIT-IBM Watson AI Lab}
  \city{Cambridge, MA}
  \country{USA}}

\author{Devavrat Shah}
\affiliation{%
  \institution{Massachusetts Institute of Technology}
  \city{Cambridge, MA}
  \country{USA}}

\renewcommand{\shortauthors}{Han, Greenewald, and Shah}

\begin{abstract}
Racial disparities in recidivism remain a persistent challenge, and the growing adoption of risk assessment algorithms has intensified the scrutiny of their sources. Past works have primarily focused on disparities in the predictions of these algorithms, viewing recidivism as a binary outcome. While sociological and criminological research has long documented non-algorithmic factors in recidivism, it remains unclear whether the risk assessments that decision-makers act on fully account for racial disparities. This work presents a multi-stage causal framework for time-to-recidivism that captures the interactions between race, the risk assessment algorithm, and contextual factors. We introduce interventional racial parity and a formal survival analysis test, conducted with observational data, of whether the algorithmic risk assessment fully accounts for racial differences in recidivism. Applied to the COMPAS dataset, the test detects no racial disparity within risk groups at short follow-up horizons. A statistically significant disparity becomes detectable after roughly nine months in the low-risk group and persists under finer score stratification, risk score perturbation, and a competing-risks analysis. This suggests that factors beyond the algorithmic scores, possibly including structural disparities in housing, employment, and social support, may shape recidivism over time, underscoring the need for policy interventions beyond algorithmic improvements, particularly for low-risk defendants.
\end{abstract}


\maketitle

\section{Introduction}\label{sec:introduction}
With millions of formerly incarcerated people returning to prison each year, the cycle of re-offending after release, known as \textit{recidivism}, remains a pressing challenge worldwide. In the United States, recidivism is a particularly complex issue closely entwined with the stark racial, economic, and social inequalities that permeate the criminal justice system. The emerging literature suggests that minority groups face disparate treatment across various stages of the criminal justice process, from the initial 911 call for service \citep{Han_Miller_Watkins_Winship_Christia_Shah_2024}, through policing \citep{fryerEmpiricalAnalysisRacial2019a}, court sentencing \citep{rehaviRacialDisparityFederal2014}, probation and parole decisions \citep{kimchiInvestigatingAssignmentProbation2019, huebnerRoleRaceEthnicity2008}, and re-entry support \citep{westernRacializedReentryLabor2019}.
Therefore, a close-up examination of the pathways and extent of such racial disparities must precede any effective reforms for a fair and equitable criminal justice system.

Amid these systemic challenges, the increasing deployment of algorithmic risk assessment tools, such as COMPAS (Correctional Offender Management Profiling for Alternative Sanctions) \citep{northpointe1996compas}, has added another layer of complexity to discussions of fairness and equity in the criminal justice system. These tools, originally designed to standardize decision-making, reduce human biases in bail, parole, and sentencing decisions, and enhance consistency and efficiency in judicial outcomes, have faced continual scrutiny and criticism since their adoption \citep{berkFairnessCriminalJustice2021, barabas2018interventions, corbett2017algorithmic}. The complexity surrounding these tools arises from three dimensions. First, it concerns the competing definitions of fairness. In particular, ProPublica’s influential news report on machine bias, which analyzed the impact of the COMPAS risk assessment tool, suggested the existence of racial disparities in terms of predictive accuracy: African-American individuals who did not recidivate within a two-year period were disproportionately labeled as higher risk compared to their Caucasian counterparts \citep{angwinMachineBias2016, larsonHowWeAnalyzed2016}. Yet this is only one of several competing definitions of fairness, which in general cannot all hold at once. Second, such an algorithm-centered analysis does not fully account for the broader structural context or potential non-algorithmic factors{, such as} the socioeconomic conditions substantially studied in sociological and criminological literature, that may also contribute to disparities in the outcomes of recidivism \citep{kubrin2006predicting, kirk2009natural, kirk2015natural}. Third, it also leaves out the time trajectory of how potential bias flows through the pathways embedded in the criminal justice system over time.

\subsection{Contributions}

The primary contribution of this work is to provide a systematic approach to answer the following question: does the algorithmic risk assessment fully account for the racial disparities in time-to-recidivism, and if not, when and how do the remaining disparities emerge? In answering this question, the key challenge lies in disentangling the interactions between algorithmic decisions, perceived race, and additional contextual factors over time. This work overcomes the limitations of prior work towards addressing this question through the contributions summarized next.

We propose a multi-stage causal framework that captures the complete trajectory from arrest to potential re-offense or return to custody. This allows us to examine both direct and indirect pathways through which racial disparities can manifest over time. The framework rests on a single modeling assumption: given fixed contexts and algorithmic decisions, race itself does not make someone recidivate sooner or later. As such, whether the algorithmic assessment fully accounts for the observed disparities is then a testable hypothesis.

Building on this framework, we introduce the notion of interventional racial (dis-)parity, a fairness criterion that examines whether individuals of different races who are otherwise identical exhibit equivalent time-to-recidivism patterns under the influence of the criminal justice system and contextual factors. We move beyond static measures of fairness that treat criminal justice outcomes as a simple True or False binary predictive question to consider the dynamic and context-dependent nature of recidivism shaped by structural inequality over time, through the lens of survival analysis.

To determine whether this racial parity holds, we establish Theorem \ref{lemma:calculation}{, which yields} a data-driven test comparing the recidivism curves across races within the same risk {score} group. The challenge {is} that the true time-to-recidivism is often masked by censoring: when an individual returns to custody for non-criminal reasons before any re-offense, the data reveal only the earlier of the two times. To address this, we {employ} the log-rank test from survival analysis, which can be properly conducted with observational censored data. The causal model itself supplies the statistical guarantee for this test: under the null hypothesis, the race labels within a risk score group are exchangeable, so the test can be calibrated by permuting them. We track the log-rank statistic across the entire follow-up period and reject when its largest deviation exceeds the corresponding permutation quantile. As such, the resulting test controls the family-wise error rate (FWER) across all follow-up horizons, that is, the probability of falsely declaring a disparity at one or more horizons, accounting for the multiple testing inherent in examining the whole trajectory (Proposition \ref{prop:exact}). As a byproduct, we obtain the earliest follow-up time at which a disparity can be declared significant, which we term the \emph{detectability onset}. Rejection of the null hypothesis therefore means precisely that the algorithmic risk assessment alone does not account for the racial disparities in time-to-recidivism, so factors beyond it contribute.

We utilize this framework to analyze the COMPAS dataset curated by ProPublica, which records the risk scores, criminal histories, and subsequent re-arrests of roughly 10,000 defendants in Broward County, Florida. Within a short-term follow-up period, we do not find sufficient evidence to support that recidivism patterns across racial groups are different within the same risk score group. However, disparities become significant once the follow-up period exceeds the \textit{detectability onset} of approximately nine months, specifically for individuals categorized as low risk by the risk assessment algorithm, thereby rejecting the null hypothesis that the algorithmic assessment alone accounts for the observed disparities. The onset also explains why analyses at different fixed horizons would reach conflicting conclusions, and hence why the choice of follow-up horizon can itself determine the policy implications of a study. The rejection is robust under additional analyses: the disparity persists within identical algorithmic decile risk scores, under deliberate perturbations of the risk scores in the direction that works against the finding, and in a competing-risks analysis that does not treat return to custody as independent censoring.

These findings have immediate social consequences, because courts and supervision agencies act directly on COMPAS risk assessments. Two defendants judged to pose the same level of risk could nevertheless face materially different chances of re-arrest that the algorithm does not capture. One possible explanation for this residual disparity is that structural inequalities in socioeconomic conditions, including disparities in access to stable housing, employment, and social support, may {shape recidivism risks} over time, {beyond} the scope of algorithmic predictions. We thus advocate for comprehensive policy interventions that look beyond algorithmic refinements, with particular attention to the low-risk population, for whom contextual conditions appear to matter most \citep{jacobsNeighborhoodRiskFactors2021}.

\subsection{Organization}

The rest of the paper is organized as follows.
In Section \ref{sec:related works}, we {review} related works on algorithmic fairness and bias in risk assessment, counterfactual fairness frameworks, and recidivism and the criminal justice system at large.
In Section \ref{sec:framework}, we lay out the {theoretical} foundation of a multi-stage causal framework to understand the pathways of racial disparities in recidivism. We introduce a data-driven test for whether factors beyond the algorithmic risk scores contribute to racial disparities in recidivism profiles, and we establish it formally.
In Section \ref{sec:empirics}, we conduct an empirical study by applying our formal framework to the COMPAS dataset and discuss candidate mechanisms and the policy implications of the observed disparities.
Section \ref{sec:conclusion} concludes.

\section{Related Works}\label{sec:related works}

The intersection of fairness and algorithmic risk assessments within the criminal justice system has drawn considerable research attention and public scrutiny. The increasing adoption of algorithmic tools such as the COMPAS risk assessment system has exhibited both the hope for data-driven decision-making and the inherent risks of embedding and amplifying existing biases. This section provides a comprehensive overview of related literature to our study, ranging from foundational analyses of algorithmic fairness to studies on recidivism disparities and causal fairness frameworks.

\textbf{Algorithmic Fairness and Bias in Risk Assessments.} The influential ProPublica investigation {reported} substantial racial disparities in the COMPAS algorithm's predictions, finding that Black defendants were more likely than White defendants to be falsely classified as high risk for recidivism despite similar reoffending rates \citep{angwinMachineBias2016, larsonHowWeAnalyzed2016}. Rigorously speaking, their main findings only compare \textit{a variant} of the two races' actual false positive and false negative rates.

Mathematically, their comparison compares the false positive rates $\mathbb{P}(M\in\{\text{medium, high}\} \mid D=d,\ \tau>2)$ and the false negative rates $\mathbb{P}(M\in\{\text{low}\} \mid D=d,\ \tau\leq 2)$ across the two values of $d$,
where $M\in\text{\{low, medium, high\}}$ denotes the algorithmic risk assessment decision, $D\in\{\majority, \minority\}$ denotes the race, and $\tau$ denotes the actual time to recidivism {in years}. However, the true time to recidivism is often masked by the time to return to custody for non-criminal violations, meaning if returning to custody happens first, then we only observe the minimum of the two, time to return to custody, instead of the target time to recidivism. This is referred to as the right-censoring problem in survival analysis, requiring more careful time-to-event examination.

 This work has also sparked debate over using \textit{equalized odds} in criminal justice settings \citep{Rudin2020Ageof, floresFalsePositivesFalse2016}. \citet{dieterichCOMPASRiskScales2016} and subsequent responses defended COMPAS's \textit{predictive parity}, i.e., $\mathbb{P}(\tau \leq 2 | D = \majority, M\in\text{\{low\}}) \simeq \mathbb{P}(\tau \leq 2 | D = \minority, M\in\text{\{low\}})$, a condition that coincides with \textit{calibration} within risk groups, the central metric in these defenses \citep{floresFalsePositivesFalse2016}, arguing that its design and operational goals inherently prioritized predictive consistency and accuracy, not necessarily equity. In fact, as shown by \citet{chouldechovaFairPredictionDisparate2016}, so long as the base rate of the two populations differs, i.e.,  $\mathbb{P}(\tau \leq 2 | D = \majority) \neq \mathbb{P}(\tau \leq 2 | D = \minority)$, \textit{equalized odds} and \textit{predictive parity} cannot hold simultaneously for any non-trivial{, imperfect} classifier.

\textbf{Causal Inference and Fairness Frameworks.} More recent works also shift from associational fairness to a causal perspective, marking a critical advancement in understanding and mitigating systemic biases. \citet{10.5555/3294996.3295162} define counterfactual fairness as a constraint on a predictor, requiring its output for an individual to remain unchanged in a counterfactual world where the individual's protected attribute differs. \citet{costonCounterfactualRiskAssessments2020} explored counterfactual fairness notions across different demographic groups, emphasizing the importance of understanding causal pathways that may drive disparities. Similarly, \citet{mishlerFairnessRiskAssessment2021b} proposed post-processing methods for achieving counterfactual equalized odds within algorithmic frameworks.
More systematically, \citet{pleckoCausalFairnessAnalysis2024} introduced a comprehensive causal toolkit for fairness analysis, highlighting the need to disentangle causal effects from correlations to identify and rectify disparities in algorithmic decision-making. Our study builds along this line of work by situating recidivism disparities within a causal context, adopting a framework that operates at the interventional level of the causal hierarchy to understand the effect of algorithmic and additional contextual factors on recidivism over time.

\textbf{Recidivism and Societal Contexts.} Recidivism is shaped by a complex interplay of individual, community, and systemic factors, including socioeconomic conditions, neighborhood characteristics, and access to essential resources. A substantial sociological and criminological literature has documented these influences and their racial patterning, including the effects of neighborhood disadvantage and residential change \citep{kubrin2006predicting, kirk2009natural, kirk2015natural} and the racial fairness of risk instruments themselves \citep{skeem2016risk}. \citet{jacobsNeighborhoodRiskFactors2021} emphasized the effects of neighborhood risk factors on recidivism, showing that neighborhood disadvantage and disorder raise recidivism mainly among people at relatively low individual risk. \citet{okonofuaScalableEmpathicSupervision2021} explored the impact of empathic supervision on recidivism reduction, emphasizing supportive interventions. \citet{goodleyPredictorsRecidivismFollowing2022} conducted a meta-analysis identifying key predictors of recidivism and their differential effects across demographic groups. Research employing survival analysis techniques has further deepened our understanding of recidivism dynamics. \citet{jungRecidivismSurvivalTime2010} investigated the association between race and recidivism among formerly incarcerated individuals, illustrating how temporal aspects of recidivism can reveal inequities that static metrics often obscure. \citet{pmlr-v219-do23a}, \citet{10.5555/3722577.3722823}, and \citet{10.1145/3534678.3539259} extend fairness analysis to general survival settings beyond recidivism, using mutual information minimization, distributionally robust optimization, and interpretable modeling, respectively, to mitigate disparities in time-to-event predictions. These works aim to develop fairer predictors, whereas our goal is to understand whether existing algorithmic risk assessment decisions fully account for disparities in time-to-recidivism. The contextual factors documented above are precisely what a risk score may fail to absorb.

\section{Unpacking Racial (Dis)parities in Recidivism: A Causal Framework}\label{sec:framework}

Recidivism is a complex and systemic issue, influenced by social, economic, and institutional factors. To understand the advent and extent of racial disparities in recidivism of individuals with comparable risk assessment profiles, we propose a multi-stage causal framework that characterizes the full trajectory, from arrest to potential reoffense or return to custody. By leveraging this framework, we examine how racial disparities emerge over time, with particular attention to the role of algorithmic risk assessments and additional contextual factors.

\subsection{Framework}

\begin{figure}
\FIGURE
{\includegraphics[width=0.6\linewidth]{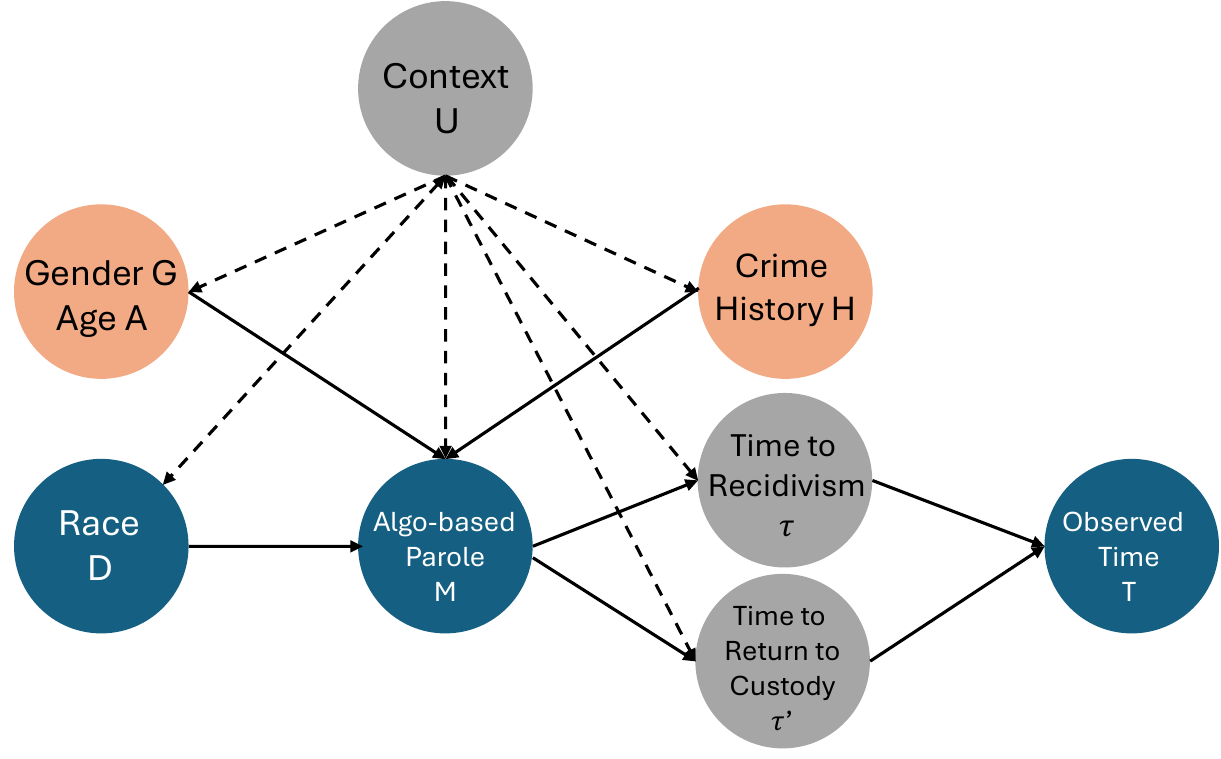}}
{A causal DAG corresponding to the multi-stage recidivism process. \label{fig:dag}}
{}
\end{figure}

We consider a cohort of arrested individuals subject to the COMPAS risk assessment tool for predicting recidivism risk. These individuals undergo varied pre-trial detention decisions and subsequent sentencing outcomes. Upon release, individuals face three possible outcomes: they may successfully reintegrate without further criminal activity, they may reoffend and be rearrested, or they may return to custody for non-criminal violations such as probation breaches. We define recidivism as our target event, measured as the time from release to rearrest. For individuals who return to custody for non-criminal reasons, we treat this return as a censoring event that masks the potential occurrence of recidivism. We formalize this multi-stage process through a natural causal framework represented by a causal Directed Acyclic Graph (DAG) in Figure \ref{fig:dag}. This structured causal framework enables us to systematically evaluate how perceived race interacts with the risk assessment algorithms, institutional decisions, and broader contexts over time.\\

\noindent
{\bf Arrested Individual.} Let $D \in \{\majority, \minority\}$ denote the race of the arrested individual.\footnote{Although extensive literature underscores the socially constructed nature of racial categories \citep{sen_race_2016}, the data constraints in large-scale recidivism studies make it challenging to adopt a fully constructivist approach. As a consequence, we follow the convention of much of the causal criminology literature, which often employs a non-constructivist perspective to align with existing studies and ensure comparability. \citet{10.1145/3351095.3375674} and \citet{10.1145/3442188.3445886} argue that treating race as a manipulable cause overlooks its constitutive social meaning. Our analysis therefore takes perceived race, rather than race itself, as the causal variable, and draws conclusions at the group level.
}

\noindent
{\bf Algorithm-based Decision.} The criminal justice system uses algorithmic risk scores to inform decisions about bail, parole, and probation, potentially shaping an individual's post-release trajectory. We use $M \in \{\text{low, medium, high}\}$, the assigned COMPAS risk score category, as a proxy for the algorithm-based criminal justice system decisions. The risk score summarizes demographic features like age, gender, race, crime history and other background information available to the system. %

\noindent
{\bf Recidivism or Returning to Custody.} Upon release, the individuals are followed up {until} they re-offend and are rearrested, return to custody for non-criminal violations{,} or the follow-up period ends, whichever comes first. Specifically, recidivism is the \textit{target} event and returning to custody or end-of-study is the \textit{censoring} event. We denote by $\tau$ the true time to recidivism, potentially unobserved in certain cases if masked by the censoring time $\tau'$. $T$ is the observed time, determined entirely by $\tau$ and $\tau'$, i.e. $T = \min\{\tau, \tau'\}$.

\noindent

\noindent{\bf Context.} Socioeconomic conditions and other contextual factors $U$ may influence multiple variables in our framework including the demographic characteristics, algorithm-based criminal justice system decision $M$, time to recidivism $\tau$, and censoring time $\tau'$.  Note that we adopt the dashed bidirectional arrow notation ($\leftarrow - -
\rightarrow$) as in \citet{pleckoCausalFairnessAnalysis2024} between protected attributes $D$ and the context $U$ to denote the associational relationship, instead of a causal one.

\noindent{\bf Causal Mechanism.} The causal relationship within this framework is governed by the interactions between context $U$, race $D$, algorithm-based criminal justice system decision $M$, and timing variables $\tau$, $\tau'$, and $T$. As noted, the context $U$ may influence $M$, $\tau$, and $\tau'$, and it is associated with $D$. %
The algorithm-based decision $M$, informed by gender, age, and crime history, may affect the observed time $T$ through the potential time-to-recidivism $\tau$ and censoring time $\tau'$. The underlying model is represented in a causal Directed Acyclic Graph (DAG) in Figure \ref{fig:dag}. A key assumption encoded in this causal mechanism is that there is {\em no direct arrow} between $D$ and $\tau, \tau'$. In words, this encodes

\begin{center}
    {\em given the societal
{\em context} and {the} algorithm-based criminal justice system decision, race does not {directly} make someone recidivate sooner or later.}
\end{center}

\subsection{Definition}\label{ssec:parity}
Given our causal framework, we now formalize the notion of racial parity by examining how the {\em intervention} of race $D$ affects time-to-recidivism under varying contexts.\footnote{The notion of manipulating a person's race is inherently controversial. One potential way to address this is by focusing on perceptions of race rather than treating it as an immutable trait \citep{Greiner_Rubin_2011}. Thus, in this work, we model a scenario where two individuals in the same context are identical except for the perception of their races and understand whether this leads to different recidivism outcomes.} Specifically, adopting the standard causal notation \citep{Pearl_2014}, we define interventional racial parity among individuals in specific algorithm-based decision groups.

\begin{definition}[Interventional Racial Parity]\label{def:1}
The criminal justice system exhibits {\em {interventional} racial parity} if {$\forall t > 0$}, $\forall m \in \{\text{low, medium, high}\}$,
\begin{align}
\p^{\doop{D=\majority}}[ \tau>t | M=m] & = \p^{\doop{D=\minority}}[\tau>t| M=m]\label{eq:parity.crime.do}.
\end{align}
\end{definition}

\noindent Here, $\p^{\doop{D = d}}(\cdot)$ represents the probability distribution under an intervention that sets race $D = d$ for $d \in \{\majority, \minority\}$.

In other words, this definition operationalizes a thought experiment: consider two individuals identical in every respect except their perceived race. Would they receive similar treatment by the criminal justice institutions and the broader context factors? Would they experience similar recidivism trajectories under the same algorithmic criminal justice system decision group? The criterion operates at the interventional level of the causal hierarchy \citep{Pearl_2014}, constraining distributions under interventions on perceived race within each decision group. More importantly, by examining this through a time-to-event lens, we can identify not just whether racial disparities exist, but when they become detectable, allowing us to better understand the temporal dynamics of how additional context (e.g. socioeconomic) factors potentially shape recidivism outcomes over time \footnote{It is important to note that there are various definitions of fairness in the literature, each suited to different contexts \citep{10.1145/3194770.3194776, 10.5555/3294771.3294834, 10.5555/3294996.3295162}. ProPublica focuses on error-rate balance, i.e., \textit{equalized odds}, while the notion of parity we adopt here{, which is derived} from \textit{outcome parity}{, is} designed to assess fairness in the time-to-recidivism outcomes across racial groups \citep{NIPS2016_9d268236}.}.

\subsection{Hypothesis} \label{sec:compare_fairness}

Having established the multi-stage causal framework and the ideal goal of achieving {interventional} racial parity, a key challenge arises from the unobservability of general contextual factors $U$, which may induce spurious associations between race $D$ and time to recidivism $\tau$, thereby making it difficult to directly assess whether {interventional} racial parity holds. To address this, we examine the role of context through hypothesis testing. Specifically, we test a necessary condition under the null hypothesis that no such spurious associations exist and use real-world data to assess whether additional contextual effects are present.

This leads us to first formulate the following hypothesis about the structural role of context and then offer a formal empirical test in the following subsection.

\begin{hypothesis}[Structural Hypothesis: Direct Effect of Context on Time-to-Recidivism]
\begin{align*}
    H_0 &: \text{context $U$ directly affects neither time-to-recidivism $\tau$ nor time-to-custody $\tau'$,} \\
    \text{vs.}\quad H_1 &: \text{context $U$ directly affects time-to-recidivism $\tau$ or time-to-custody $\tau'$.}
\end{align*}
\end{hypothesis}

Under $H_0$, according to the causal DAG in Figure \ref{fig:dag}, we can represent the causal quantity $\p^{\doop{D=d}}[ \tau>t| M=m]$ as follows:

\begin{theorem}\label{lemma:calculation}
Under $H_0$ that the context $U$ does not directly affect time-to-recidivism $\tau$ and time-to-custody $\tau'$, $\forall t>0, m \in \{\text{low, medium, high}\}$,
\begin{align}
\p^{\doop{D=d}}[ \tau>t | M=m] &= \p[ \tau>t| D=d, M=m]=\p[ \tau>t|M=m].\label{eq:do.2}
\end{align}
\end{theorem}

\begin{proof}
Under $H_0$, we remove the edges $U \to \tau$ and $U \to \tau'$ from Figure \ref{fig:dag}. Thus, $D$ and $\tau$ are d-separated by $M$, i.e. $D\perp \tau |M$.

Based on the modified DAG and do calculus, we have $\forall d \in \{\majority, \minority\}, t>0, m \in \{\text{low, medium, high}\}$
\begin{align}
   & \p^{\doop{D=d}}[ \tau>t| M=m] \nonumber\\
   =& \sum_u \p^{\doop{D=d}}[ \tau>t, U=u, D=d| M=m]\nonumber\\
   =& \sum_u \p^{\doop{D=d}}[ \tau>t, U=u | D=d, M=m] \cdot
  \p^{\doop{D=d}}[D=d | M=m]\label{eq:products}\\
  =& \sum_u \p[ \tau>t, U=u | D=d, M=m]\label{eq:one}\\
  =& \p[ \tau>t| D=d, M=m] \label{eq:conditional}\\
  =& \p[ \tau>t| M=m] \label{eq:final}
\end{align}
where \eqref{eq:one} is obtained from \eqref{eq:products} due to $\p^{\doop{D=d}}[D=d | M=m]=1$ and \eqref{eq:final} is obtained from \eqref{eq:conditional} due to $D\perp \tau |M$.
\end{proof}

\begin{corollary}\label{cor:1}
    Under $H_0$ that the context $U$ does not directly affect time-to-recidivism $\tau$ and time-to-custody $\tau'$, $\forall t>0, m \in \{\text{low, medium, high}\}$, we have {interventional} racial parity, i.e. $\p^{\doop{D=\majority}}[ \tau>t | M=m] = \p^{\doop{D=\minority}}[ \tau>t | M=m]$.
\end{corollary}

\begin{proof}
    Since Theorem \ref{lemma:calculation} holds for $\forall d \in \{\majority, \minority\}$, we have $\p^{\doop{D=\majority}}[ \tau>t | M=m] = \p[ \tau>t| M=m] = \p^{\doop{D=\minority}}[ \tau>t | M=m]$.
\end{proof}

\textbf{Key Implication.}
Theorem \ref{lemma:calculation} thus allows the interventional quantity to be estimated from observational data: under $H_0$, it equals the ordinary survival probability within the risk group, regardless of race. The intuition behind this and Corollary \ref{cor:1} is that when the risk categories fully explain the disparities and additional contextual factors have no direct effect on recidivism timing, controlling for algorithmic decisions alone should ensure interventional racial parity.

Moreover, the contrapositive argument of Theorem \ref{lemma:calculation} leads to a practical test. If we observe different recidivism patterns across racial groups within the same algorithmic decision category, we can reject $H_0$. Such a finding provides evidence that the algorithmic risk categories alone do not fully account for the observed disparities, and is consistent with the presence of additional contextual factors that affect time-to-recidivism or time-to-custody. More formally, fix $m \in \{\text{low, medium, high}\}$ and let $\phi$ be a level-$\alpha$ test of the null hypothesis $H_0(m): \p[ \tau>t| D=\majority, M=m] = \p[ \tau>t| D=\minority, M=m]$ for all $t>0$. If $\phi$ rejects, then additional contextual factors are present. We explain below how to perform such an evaluation when the actual time-to-recidivism can be masked due to censoring.

\subsection{Empirical Test}\label{sec:empirical_test}
One potential challenge in directly using Theorem \ref{lemma:calculation} is that we often cannot observe the true time-to-recidivism $\tau$ for all individuals, only a lower bound of $\tau$. This occurs because some individuals return to custody for non-criminal violations, like missing probation meetings, before any potential reoffense, a phenomenon known as censoring in survival analysis. Traditional statistical tests that ignore censoring could produce biased results, as mistakenly using time-to-custody {would} underestimate the true time-to-recidivism. Survival analysis methods are specifically designed to handle such censored data by properly accounting for both observed recidivism events and censored observations, thereby allowing us to decide if we have enough empirical evidence to reject the null hypothesis $H_0$ or not.

To formalize this testing procedure, we state the following empirical test.

\begin{test}
Let $S_{d}(t \mid m) := \p[\tau>t \mid D=d, M=m]$ for $d \in \{\majority, \minority\}$. For each $m \in \{\text{low, medium, high}\}$, we test
\begin{align*}
H_0(m):\quad
& S_{\majority}(t \mid m)
=
S_{\minority}(t \mid m)
\qquad \text{for all } t>0,\\
\text{vs.}\quad
H_1(m):\quad
& S_{\majority}(t \mid m)
\neq
S_{\minority}(t \mid m)
\qquad \text{for some } t>0.
\end{align*}
\end{test}

Under the null hypothesis, individuals of different races but the same algorithmic risk score group should have identical survival curves; that is, their probability of remaining arrest-free should be the same at all time points. To test this hypothesis while properly accounting for censoring, we employ the non-parametric log-rank test, which compares the entire survival curve rather than outcomes at a single time point.

\textbf{Converting Recidivism Data to Survival Data Format.} For each individual $i$, the observational recidivism data collects their algorithmic risk score group $M_i$, race $D_i$, observed time $T_i$, {and} whether re-arrest occurs $\Delta_i$. $\forall m \in \{\text{low, medium, high}\}$, we track four key quantities at each time point $t$ which can be converted from observational data:
\begin{itemize}
    \item $O_{d,m,t}$: Number of observed re-arrests of race $d$ in risk assessment group $m$ at time $t$, i.e. $|\{i| D_i=d, M_i = m, T_i =t, \Delta_i = 1\}|$
    \item $N_{d,m,t}$: Number of individuals of race $d$ in risk assessment group $m$ who have not been rearrested or returned to custody at time $t$, i.e. $|\{i| D_i=d, M_i = m, T_i \geq t\}|$
    \item $O_{m, t}$: Total number of observed re-arrests of both races in risk assessment group $m$ at time $t$,  i.e. {$|\{i \mid M_i = m, T_i =t, \Delta_i = 1\}|$}
    \item $N_{m,t}$: Total number of individuals of both races in risk assessment group $m$ who have not been rearrested or returned to custody at time $t$, i.e. {$|\{i \mid M_i = m, T_i \geq t\}|$}
\end{itemize}
The expected number of re-arrests for each race $d$ in risk assessment group $m$ at time $t$, assuming $H_0(m)$, is calculated as $E_{d, m, t}=O_{m, t}\cdot \frac{N_{d,m,t}}{N_{ m, t}}$.

\textbf{Test Statistic.} For each horizon $t$, the standardized log-rank statistic compares observed and expected re-arrests of the majority group up to $t$,
\begin{align*}
z_m(t) = \frac{\sum_{s \le t} \left(O_{\majority,m,s} - E_{\majority,m,s}\right)}{\sqrt{\sum_{s \le t} \frac{N_{\majority,m,s}N_{\minority,m,s}O_{m,s}(N_{m,s}-O_{m,s})}{N_{m,s}^2(N_{m,s}-1)}}},
\end{align*}
where the sums run over event times $s$ up to $t$. At the full follow-up horizon, $z_m(t)^2$ is the classical log-rank statistic, which under $H_0(m)$ is asymptotically chi-square distributed with one degree of freedom. The sign of $z_m(t)$ indicates which racial group is rearrested more often than expected under $H_0(m)$.

\textbf{Assumptions of Log-rank Test.} The log-rank test of $H_0(m)$ is valid under two conditions: (1) Independent observations: the pairs $(T_i, \Delta_i)$ are independent across individuals. (2) Independent censoring: within each race and risk score group, the time to return to custody is independent of the time to recidivism, i.e., $\tau \perp \tau' \mid D=d, M=m$ for each $d$. Both conditions are implied by the causal model under $H_0$ %
and guarantee the validity of the test. Once the test rejects, however, $H_0$ no longer holds and independent censoring is no longer guaranteed, so to describe the disparity itself Section \ref{sec:empirics} also reports a competing-risks comparison, which treats rearrest and return to custody as competing events and lifts the assumption about how the two are related.

\textbf{Inference Across All Follow-up Horizons.}
To determine when a disparity emerges, we examine $z_m(t)$ over all horizons. To account for multiple hypothesis testing, we base the test on the maximal deviation over the follow-up period and calibrate it by randomly permuting the race labels within the risk score group, as formalized in Definition \ref{def:perm}.

\begin{definition}[Permutation Supremum Test]\label{def:perm}
Fix a risk score group $m$ and let $W_m = \sup_t |z_m(t)|$. Let $W_m^{(1)}, \dots, W_m^{(B)}$ denote the values of this statistic recomputed after $B$ independent, uniformly random permutations of the race labels within group $m$. For $\alpha \in (0,1)$, let $c^*_\alpha$ be the $\lceil (1-\alpha)(B+1) \rceil$-th smallest value among $W_m, W_m^{(1)}, \dots, W_m^{(B)}$. The test rejects $H_0(m)$ when $W_m > c^*_\alpha$. Its \emph{horizon-adjusted $p$-value}, which accounts for all follow-up horizons within the risk group, is $\left(1 + \#\{b : W_m^{(b)} \ge W_m\}\right)/(B+1)$. The test declares a disparity at every horizon $t$ with $|z_m(t)| > c^*_\alpha$. The \emph{detectability onset} is the earliest such horizon.
\end{definition}

\begin{proposition}[Finite-Sample Validity]\label{prop:exact}
Suppose the observations are independent draws from the causal model of Figure \ref{fig:dag} and that $H_0$ holds. Then, for any $\alpha \in (0,1)$ and any number of permutations $B$, the test in Definition \ref{def:perm}
(i) rejects $H_0(m)$ with probability at most $\alpha$, and
(ii) controls the family-wise error rate (FWER) at level $\alpha$, that is, the probability of falsely declaring a disparity at one or more horizons is at most $\alpha$.
\end{proposition}

\begin{proof}
Under $H_0$ the edges $U \to \tau$ and $U \to \tau'$ are absent from Figure \ref{fig:dag}, so $D \perp (\tau, \tau') \mid M$. The pair $(T, \Delta) = (\min\{\tau, \tau'\}, \mathds{1}\{\tau \le \tau'\})$ is a function of $(\tau, \tau')$, hence $D \perp (T, \Delta) \mid M$. Since individuals are independent draws, conditional on the observed pairs $(T_i, \Delta_i)$ and the number of individuals of each race in group $m$, all assignments of the race labels to the pairs are equally likely. The observed labeling is therefore exchangeable with the $B$ permuted ones, so the rank of the observed $W_m$ among the $B+1$ values is uniformly distributed. The observed $W_m$ exceeds $c^*_\alpha$ only if its rank is above $\lceil (1-\alpha)(B+1) \rceil$, which happens with probability at most $\alpha$, with ties only lowering this probability. This proves (i). For (ii), a false declaration at any horizon requires $|z_m(t)| > c^*_\alpha$ for some $t$, an event contained in $\{W_m > c^*_\alpha\}$, whose probability is at most $\alpha$ by (i).
\end{proof}

Proposition \ref{prop:exact} holds in finite samples and it relies neither on the large-sample chi-square approximation nor on a model for the censoring mechanism, because the exchangeability of the labels follows directly from the causal null hypothesis. For interpretation, we express $c^*_\alpha$ on the $p$-value scale as $\alpha^* = 2(1-\Phi(c^*_\alpha))$, the threshold against which the pointwise $p$-values of $z_m(t)$ can be read. The detectability onset is thus the earliest time at which a disparity can be declared while the FWER over the entire follow-up period remains controlled.

\section{Empirics}\label{sec:empirics}

Having developed a causal framework and a formal empirical test for analyzing racial disparities in recidivism, in this section, we use the COMPAS dataset collected by ProPublica to evaluate whether factors beyond the algorithmic risk scores contribute to the observed disparities in time-to-recidivism. Our causal framework suggests that if disparities persist even after controlling for algorithmic risk scores, this would indicate the presence of additional influences on recidivism trajectories.
Specifically, we apply the inference protocol developed in Section \ref{sec:empirical_test} to examine whether and when racial disparities emerge in time-to-recidivism patterns. This allows us to assess not just the existence of disparities, but also their temporal dynamics, whether disparities appear immediately post-release or develop over longer follow-up periods, and how early they can be detected.

\subsection{Data Description}
The COMPAS dataset, curated by ProPublica in 2016, contains demographic characteristics, criminal history, and COMPAS risk scores of around 10,000 criminal defendants in Broward County, Florida. This dataset includes all individuals who underwent COMPAS assessments at the pretrial stage in 2013 and 2014 and had public criminal records up to April 1, 2016 for tracking subsequent offenses. Following prior analyses in {the} literature, we consider \textit{{Caucasian}} as the \textit{majority} race and \textit{African-American} as the \textit{minority} race. Score groups follow the standard COMPAS mapping of decile scores: low (1 to 4), medium (5 to 7), and high (8 to 10).

The public files record, for each defendant, one or more out-of-custody intervals measured in days since the COMPAS screening. Our analysis cohort takes each defendant's first release interval, so that every defendant contributes one statistically independent record: follow-up runs from release until rearrest for a new offense (the \textit{target} event) or censoring, which occurs at return to custody without a new offense or at the end of data collection near April 1, 2016. Overall, the cohort contains 8,690 defendants, of whom 2,410 were rearrested and 6,280 were censored, with a median potential follow-up of 27 months and a maximum of 39 months, measured from release to the end of data collection. Table \ref{tab:cohort} reports the composition by risk group and race.

\begin{table}[h]
\TABLE
{Cohort composition by COMPAS recidivism risk group and race. \label{tab:cohort}}
{\begin{tabular}{llrrr}
\hline
Risk group & Race & $N$ & Rearrests & Censored \\
\hline
Low & Caucasian & 2{,}362 & 409 & 1{,}953 \\
Low & African-American & 2{,}179 & 492 & 1{,}687 \\
Medium & Caucasian & 799 & 252 & 547 \\
Medium & African-American & 1{,}538 & 547 & 991 \\
High & Caucasian & 396 & 147 & 249 \\
High & African-American & 1{,}416 & 563 & 853 \\
\hline
\end{tabular}}
{}
\end{table}

It is important to note that the COMPAS dataset, while widely used, has limitations including potential sampling biases, variations in law enforcement practices, and the absence of certain contextual factors such as socioeconomic status or access to community support. Additionally, the COMPAS dataset reflects only a specific jurisdiction, Broward County, Florida, which may limit generalizability to other regions with differing criminal justice practices.

\subsection{Results}

\begin{figure}[h!]
\FIGURE
{\includegraphics[width=\linewidth]{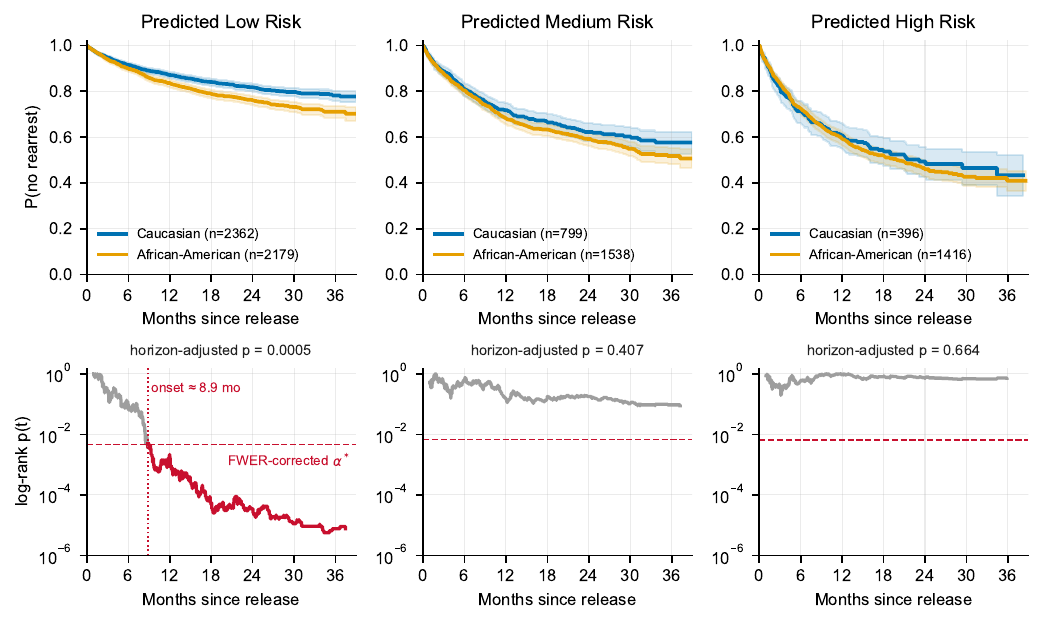}}
{No-rearrest curves with 95\% confidence intervals by race within COMPAS recidivism risk groups (top row) and the pointwise log-rank $p$-value trajectory (bottom row). The red dashed line marks the family-wise corrected threshold $\alpha^*$ from the permutation supremum test of Proposition \ref{prop:exact}; panel titles report the horizon-adjusted $p$-values of Definition \ref{def:perm}; the dotted vertical line marks the detectability onset. \label{fig:results}}
{}
\end{figure}

To examine potential racial disparities in recidivism patterns, we conducted survival analyses stratified by COMPAS risk groups. Figure \ref{fig:results} presents, for each risk group, the no-rearrest curves by race together with the pointwise $p$-value trajectory and its family-wise corrected threshold $\alpha^*$. All permutation quantities use $B = 2{,}000$ label permutations within each risk group. The primary test in each risk group is the permutation supremum test in Definition \ref{def:perm}, which already accounts for testing at every follow-up horizon. We report its horizon-adjusted $p$-value and use $\alpha = 0.05$ throughout. The robustness analyses below use the same supremum construction wherever they track a statistic over follow-up time.

\textbf{Temporal structure in primary tests.} The test rejects for the low-risk group (horizon-adjusted $p = 5\times10^{-4}$, the smallest value attainable with $B = 2{,}000$ permutations), where African-American defendants are rearrested sooner and more often than Caucasian defendants with the same assessment. The medium- and high-risk groups show no rejection (horizon-adjusted $p = 0.41$ and $0.66$), so the evidence concentrates in the low-risk population, the largest group in the data. Within the low-risk group, Figure \ref{fig:results} also shows when the disparity becomes detectable. The trajectory of pointwise $p$-values first crosses the FWER-corrected threshold $\alpha^* \approx 0.0047$ at day 270, so the detectability onset is approximately nine months.

\textbf{Robustness.} We test whether the low-risk finding survives three alternative explanations. %

The first check asks whether the gap arises because the low-risk group pools four decile scores, so that the two races could differ in their distributions of decile scores. To rule this out, we use a stratified log-rank test, which compares the races within each decile score in a risk group and later pools these comparisons across the deciles within that risk group. As in the primary test, we track this statistic over all follow-up horizons and calibrate its maximal deviation by permuting the race labels within each decile score. The test rejects for the low-risk group (horizon-adjusted $p = 0.033$), while the corresponding tests for the medium- and high-risk groups do not. The gap therefore does not come from differences in the distribution of decile scores within the low-risk group.

The second check asks how robust the gap is with respect to potential errors in the risk scores. We thus randomly perturb the scores in the direction that works against the original finding, moving a random 10\% of African-American defendants from the low to the medium group and a random 10\% of Caucasian defendants from the medium to the low group, and repeat this 100 times. The permutation supremum test rejects in every draw, with the largest horizon-adjusted $p$-value equal to $0.013$, thereby surviving potential score errors of this size.

Finally, the third check concerns the modeling of the censoring event, return to custody. Even though Proposition \ref{prop:exact} itself does not depend on how censoring operates, interpreting the curves in Figure \ref{fig:results} as times to rearrest requires return to custody to be independent of rearrest. To avoid treating return to custody as independent censoring, we instead model it as a competing event and compare the cumulative incidence of rearrest, the observed probability of being rearrested by time $t$ before any return to custody. We estimate it with the Aalen-Johansen estimator \citep{aalen1978empirical} and test with the permutation supremum procedure in Definition \ref{def:perm} again. In the low-risk group, the cumulative incidence of rearrest is up to 6.5 percentage points higher for African-American defendants (horizon-adjusted $p = 5\times10^{-4}$), while the cumulative incidence of return to custody shows no racial difference in any risk group (Figure \ref{fig:cif}). The disparity therefore appears directly in observed rearrests, and return-to-custody decisions show no detectable racial difference within the same risk group.

\begin{figure}[h]
\FIGURE
{\includegraphics[width=\linewidth]{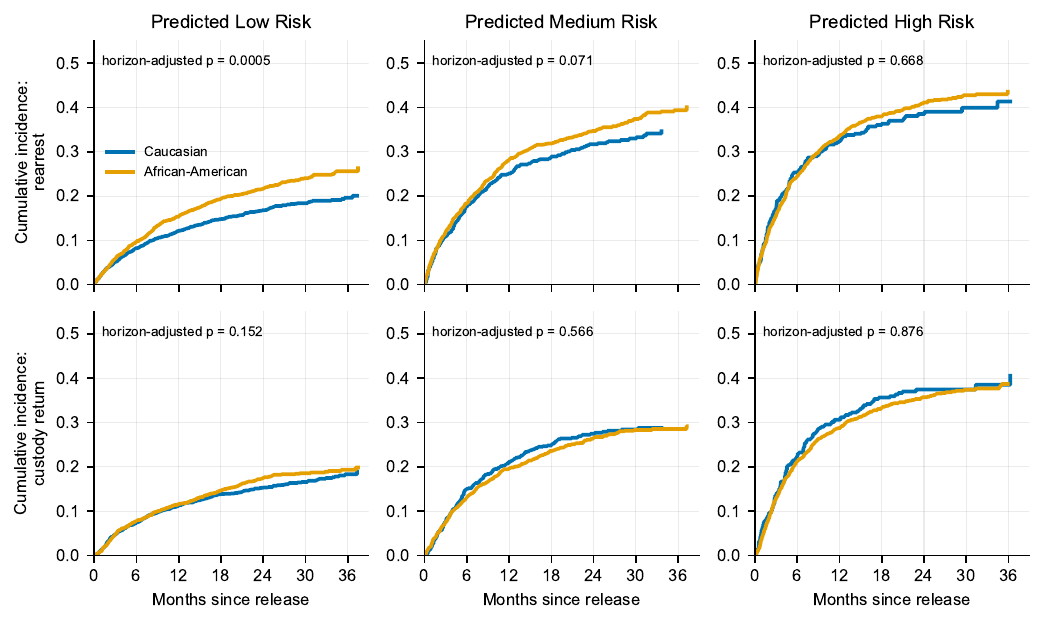}}
{Cause-specific cumulative incidence functions by race within COMPAS recidivism risk groups: rearrest (top row) and return to custody without a new offense (bottom row), with horizon-adjusted $p$-values for the supremum difference between races. The rearrest incidence differs significantly in the low-risk group while the custody-return incidence differs in none. \label{fig:cif}}
{}
\end{figure}

\textbf{Discussion.} Under the causal model of Section \ref{sec:framework}, the rejection for the low-risk group means that $H_0$ fails, so contextual factors $U$ directly affect time to rearrest or time to return to custody. The characteristics recorded in the data do not explain the gap. Adjusting for additional demographic and charge characteristics in a Cox proportional hazards model \citep{cox1972regression} leaves the gap largely intact, with a hazard ratio of 1.25 for African-American relative to Caucasian defendants in the low-risk group (95\% confidence interval $[1.09, 1.43]$). Among the unrecorded contextual factors, one plausible candidate is unequal access to post-release support, such as stable housing, employment, and social services. \citet{castilloRecidivismBarriersReintegration2024} find that encouraging the use of reentry support services reduces rearrest, with the largest benefit among Black participants, which suggests that such support matters most for the group facing the disparity here. Other contextual channels, such as differences in post-release enforcement intensity, may also contribute. The concentration of the disparity in the low-risk group is also consistent with prior evidence. \citet{jacobsNeighborhoodRiskFactors2021} find that neighborhood disadvantage and disorder increase recidivism among people with relatively low individual risk but not among those at higher risk, suggesting that contextual conditions matter most for the low-risk population. This population therefore offers a clear opportunity for policy, since the reentry support services discussed above could be extended to low-risk defendants who would otherwise receive little support and recidivate sooner. The detectability onset of roughly nine months has a separate implication for future study designs. Evaluations of risk instruments and reentry programs should follow releasees well beyond the first half year, since analyses truncated earlier lack the power to surface disparities that may potentially evolve with time.

\section{Conclusion}\label{sec:conclusion}

This study presents a comprehensive examination of racial disparities in recidivism through a multi-stage causal framework, focusing on the interactions between algorithmic risk assessments, {release} decisions, and broader contextual factors. While prior work on algorithmic risk assessment has largely focused on static measures of algorithmic bias, our analysis reveals the critical importance of temporal dynamics in understanding and addressing disparities in criminal justice outcomes. Through an exact inference procedure applied to the COMPAS dataset, we establish that defendants who share the same low-risk assessment experience systematically different rearrest trajectories by race, that the gap persists among defendants with identical decile scores, and that roughly nine months of follow-up are required before the disparity becomes statistically demonstrable. One possible explanation for {the residual gap} is the impact of societal factors, such as unequal access to housing, employment opportunities, and social support systems. Due to data limitations, in this manuscript we cannot directly verify whether the socioeconomic factors are indeed the additional sources of {disparity}. Linking recidivism records to data on housing, employment, and support services would allow a direct test of these channels, and applying the framework to other jurisdictions would assess the generality of the finding.

Policymakers and practitioners should recognize that a risk instrument whose scores are well calibrated on average can still leave systematic differences among defendants it assesses alike, so fairness auditing should extend beyond the algorithm itself to the outcomes that follow its assessments. Effective interventions may need to extend beyond algorithmic refinements to encompass policy changes that promote equitable access to housing, stable employment, and community-based support services. Such efforts can contribute to fostering a more just and equitable criminal justice system.

Our framework has applicability beyond the study of recidivism. It can be extended to other domains where algorithmic decisions intersect with social inequality. For instance, a similar causal framework can be utilized to evaluate racial bias in loan approval and default. In particular, the approved loan amount can be considered as a proxy for repayment ability, and time to repayment {as the} target event, often masked by time to default. In this scenario, socioeconomic conditions like employment stability and family circumstances also exert great influence on how soon people repay their loan. This broader applicability suggests the value of examining time-varying disparities across various institutional contexts where algorithmic systems are deployed.

Ultimately, this work underscores that fairness cannot be judged from the algorithm alone. The decisions its scores inform and the conditions people face after release also shape outcomes, and they deserve the same scrutiny. As society increasingly relies on algorithmic tools for high-stakes decisions, frameworks that can capture these complex temporal and contextual dynamics become essential for advancing both algorithmic fairness and social justice.

\bibliographystyle{ACM-Reference-Format}
\bibliography{yourbibfile}

\end{document}